\pdfoutput=1
\documentclass[aps,pra,amsmath,amssymb,twocolumn,10pt]{revtex4-2}
\usepackage{mathrsfs}
\usepackage[T1]{fontenc}
\fontencoding{T1}  
\usepackage[utf8]{inputenc}
\usepackage{enumitem}
\usepackage{tikz}

\usepackage{changes}

\newcommand{\vertiii}[1]{{\left\vert\kern-0.25ex\left\vert\kern-0.25ex\left\vert #1 
    \right\vert\kern-0.25ex\right\vert\kern-0.25ex\right\vert}}

\usepackage{amsfonts}
\usepackage{amsmath,amssymb,amsthm}
\usepackage{graphicx}
\usepackage{fancyhdr}
\usepackage[breakable]{tcolorbox}
\usepackage{bbm}
\usepackage{url}
\usepackage{mathtools}
\usepackage{braket}
\usepackage{upgreek }
\usepackage{dsfont}
\usepackage{collectbox}
\usepackage{braket}
\usepackage{quantikz}

\usepackage[plain]{fancyref} 
\usepackage[ruled]{algorithm2e}

\usepackage{algpseudocode}

\usepackage{hyperref}

\newtheorem{thm}{Theorem}
\newtheorem{coro}{Corollary}
\newtheorem*{thm*}{Theorem}
\makeatletter
\newcommand{\setthmtag}[1]{
  \let\oldthethm\thethm
  \newcommand{\thethm}{#1}
  \g@addto@macro\endthm{
    \addtocounter{thm}{-1}
    \global\let\thethm\oldthethm}
  }
\makeatother

\newtheorem*{prop*}{Proposition}
\newtheorem{lemma}[thm]{Lemma}
\newtheorem*{lemma*}{Lemma}
\newtheorem{cor}[thm]{Corollary}
\newtheorem*{cor*}{Corollary}

\newtheorem*{cj*}{Conjecture}

\newtheorem*{Def*}{Definition}

\theoremstyle{definition}

\newtheorem*{rem*}{Remark}

\def\beq{\begin{equation}}
\def\eeq{\end{equation}}
\def\bq{\begin{quote}}
\def\eq{\end{quote}}
\def\ben{\begin{enumerate}}
\def\een{\end{enumerate}}
\def\bit{\begin{itemize}}
\def\eit{\end{itemize}}

\def\lb{\left(}
\def\rb{\right)}

\def\r|{\right|}

\def\i{\text{id}}

\newcommand{\ketbra}[2]{|#1\rangle\langle #2|}
\newcommand{\tr}[1]{\operatorname{tr}\lb#1\rb}

\newcommand{\norm}[1]{\left\|#1\right\|}

\newcommand\be{\begin{equation}}
\newcommand\ee{\end{equation}}

\begin{document}
\title{Optimal quantum algorithm for Gibbs state preparation}

\author{\begingroup
\hypersetup{urlcolor=navyblue}
\href{https://orcid.org/0000-0001-7712-6582}{Cambyse Rouz\'{e}
\endgroup}
}
\email[Cambyse Rouz\'{e} ]{rouzecambyse@gmail.com}
 \affiliation{Inria, Télécom Paris - LTCI, Institut Polytechnique de Paris, 91120 Palaiseau, France}

\author{\begingroup
\hypersetup{urlcolor=navyblue}
\href{https://orcid.org/0000-0001-9699-5994}{Daniel Stilck Fran\c{c}a}
\endgroup}
\email[Daniel Stilck Fran\c ca ]{dsfranca@math.ku.dk}
\affiliation{Univ Lyon, ENS Lyon, UCBL, CNRS, Inria, LIP, F-69342, Lyon Cedex 07, France}
\affiliation{Department of Mathematical Sciences, University of Copenhagen, Universitetsparken 5, 2100 Denmark}

\author{\begingroup
\hypersetup{urlcolor=navyblue}
\href{https://orcid.org/0000-0002-5889-4022}{\'Alvaro M. Alhambra
\endgroup}
}
\email[\'Alvaro M. Alhambra]{alvaro.alhambra@csic.es}
 \affiliation{Instituto de F\'{i}sica T\'{e}orica UAM/CSIC, C. Nicol\'{a}s Cabrera 13-15, Cantoblanco, 28049 Madrid, Spain}

\begin{abstract}

It is of great interest to understand the thermalization of open quantum many-body systems, and how quantum computers are able to efficiently simulate that process. A recently introduced dissispative evolution, inspired by existing models of open system thermalization, has been shown to be efficiently implementable on a quantum computer. Here, we prove that, at high enough temperatures, this evolution reaches the Gibbs state in time scaling logarithmically with system size. The result holds for Hamiltonians that satisfy the Lieb-Robinson bound, such as local Hamiltonians on a lattice, and includes long-range systems. To the best of our knowledge, these are the first results rigorously establishing the rapid mixing property of high-temperature quantum Gibbs samplers, which is known to give the fastest possible speed for thermalization in the many-body setting. We then employ our result to the problem of estimating partition functions at high temperature, showing an improved performance over previous classical and quantum algorithms.
\end{abstract}

\maketitle


Systems in nature are most often coupled to an external bath at a fixed temperature, which naturally induces a thermalization process. The way this process is typically modelled in quantum systems is through a Lindblad master equation, which can be derived from first principles under standard approximations \cite{Rivas_2012}. 
For quantum many-body systems, however, these standard approximations do not hold in general, and different schemes need to be devised (see e.g. \cite{Taj_2008,Mozgunov_2020,Nathan_2020} for different proposals). 

The difficulty of constructing master equations for quantum many-body systems has also hindered the development of quantum algorithms reproducing the thermalization process. Fortunately, recent works \cite{chen2023quantumthermalstatepreparation,chen2023efficient} have constructed master equations that both resemble what is believed to be that thermalization process, while also being simulable on a quantum computer \cite{ding2024efficientquantumgibbssamplers,gilyen2024quantumgeneralizationsglaubermetropolis}. This gives us a method to study the physics of open system thermalization \cite{kastoryano2024littlebitselfcorrection,dabbruzzo2024steadystateentanglementscalingopen} while also potentially allowing us to efficiently prepare Gibbs states on a quantum computer, in analogy with the classical Monte Carlo methods \cite{levin2017markov}. 

One of the most important questions is then: how fast is that thermalization process? Answering it is far from straightforward: We expect that there are instances with very long thermalization times, typically at low temperature, while at the same time for a wide range of models and regimes the process should be rather quick \cite{martinelli1999lectures,kastoryano2016quantum,bardet2023rapid,bardet2024entropy}. The fastest possible way for it to occur is called \emph{rapid mixing}. In systems of $n$ particles this means that, for all initial states $\rho$, there exists a constant $\gamma >0$ such that
\begin{align}\label{equ:rapid_mixing}
   \|e^{t\mathcal{L}^{(\beta)}}(\rho)-\sigma_\beta\|_1 \le \text{poly}(n) e^{-\gamma t},
\end{align}
where $\norm{\cdot}_1$ is the trace norm.
That this convergence cannot be quantitatively improved follows from the trivial example of $1$-local dissipation channels, such as depolarizing noise \cite{depolarizing2016}. Previously, rapid mixing was only known to hold for the thermalizing dynamics of commuting Hamiltonians \cite{Majewski1995,temme2015fast,bardet2023rapid,bardet2024entropy,kochanowski2024rapidthermalizationdissipativemanybody}.

As our main result, we show that rapid mixing holds for all lattice Hamiltonians above a critical temperature $\beta^*$, which we can bound explicitly in terms of the lattice dimension and the Lieb-Robinson velocity.
This rapid mixing, together with the quantum algorithm for implementing $e^{t\mathcal{L}^{(\beta)}}$ from \cite{chen2023efficient}, implies that Gibbs states above that temperature can be prepared efficiently on a quantum computer in $\widetilde{\mathcal{O}}(n)$-time.
In comparison, previous work on the thermalization speed of Gibbs samplers by the authors \cite{rouze2024efficient} did not achieve rapid mixing, since it only estimated the gap of the generator. A gap estimate only implies we need to simulate the Hamiltonian for $\textrm{poly}(n)$ time to converge, whereas Eq.~\eqref{equ:rapid_mixing} implies that $\log(n)$ simulation time suffices. Moreover, explicit bounds on the threshold temperature were missing in \cite{rouze2024efficient} and the proof did not cover long range interactions, as it now does.

As an application, we show how to use the Gibbs sampling algorithm as a subroutine in the efficient estimation of partition functions at high temperatures. For the case of finite range Hamiltonians at high temperatures, we achieve a polynomial speed-up over classical methods based on cluster expansions \cite{Harrow2020b,Mann2021,Haah_2024}. For the case of long range interactions, to the best of our knowledge, our work gives the first result on the efficient estimation of quantum partition functions.

\medskip

\emph{Quantum Gibbs sampling—} We start by introducing the generator of the quantum Gibbs sampling algorithm from \cite{chen2023efficient}. We set a $\beta>0$, and a Hamiltonian $H$ over a $D$-dimensional lattice of spin systems $\Lambda\equiv [0,L]^D$ with system size $n:=|\Lambda|=(L+1)^D$. The Lindbladian is
\begin{align}\label{eq:generator1}
\mathcal{L}^{(\beta)}(\rho)&=-i[B,\rho]+\sum_{a\in \Lambda,\alpha\in[3]}\int_{-\infty}^{\infty}\gamma(\omega)\mathcal{D}^{a,\alpha}_\omega(\rho)\,d\omega\\
&\equiv \sum_{a\in\Lambda}\mathcal{L}_a^{(\beta)}(\rho)\nonumber\,,
\end{align}
where $\gamma(\omega):=\operatorname{exp}(-{(\beta\omega+1)^2}/{2})$ and $\mathcal{D}^{a,\alpha}_\omega$ is a dissipative term with jump operators
\begin{align*}
A^{a,\alpha}(\omega):=\frac{1}{\sqrt{2\pi}}\int_{-\infty}^{\infty} e^{iHt}A^{a,\alpha} e^{-iH t}e^{-i\omega t}\,f(t)\,dt\quad 
\end{align*}
with  $f(t):=\operatorname{exp}(-{t^2}/{\beta^2})\,\sqrt{\beta^{-1}\sqrt{2/\pi}}\,$.
We define the local Pauli matrices on site $a$ as $A^{a,1}\equiv \sigma_x,\,A^{a,2}\equiv \sigma_y$ and $A^{a,3}\equiv \sigma_z$. The Hamiltonian term is given by
\begin{align*}
&B=\sum_{a\in \Lambda,\alpha\in [3]}\, \int_{-\infty}^\infty \beta^{-2}b_1(t/\beta) e^{-i Ht}  \\
&\qquad \times\int_{-\infty}^\infty b_2(t'/\beta) e^{iHt'}A^{a,\alpha} e^{-2i Ht'}A^{a,\alpha} e^{i Ht'}dt' e^{i Ht}dt
\end{align*}
where $b_1,b_2$ are fast-decaying functions with $\|b_1\|_{L_1}, \|b_2\|_{L_1}\le 1$, defined as
\begin{align}\label{eq:b1b2}
&b_1(t):= 2\sqrt{\pi}\, e^{\frac{1}{8}}\, \left(\frac{1}{\operatorname{cosh}(2\pi t)}\ast_t\,\sin(-t)e^{-2t^2}\right),\\
&b_2(t):=\frac{1}{2\pi}\,\sqrt{\frac{1}{\pi}}\,\operatorname{exp}\big(-4t^2-2it\big)\,.
\end{align}
Given that the functions $f(t),b_i(t)$ decay rapidly, the Lindbladian is quasi-local. At the same time, it is such that in the Heisenberg picture
\begin{align*}
\langle X,\mathcal{L}^{(\beta)\dagger}(Y)\rangle_{\sigma_\beta}=\langle \mathcal{L}^{(\beta)\dagger}(X),Y\rangle_{\sigma_\beta}
\end{align*}
where $\langle X,Y\rangle_{\sigma_\beta}:=\operatorname{tr}(X^\dagger\sigma_\beta^{\frac{1}{2}}Y\sigma_\beta^{\frac{1}{2}})$ denotes the KMS inner product. This is the property of \emph{quantum detailed balance}, which guarantees that the Gibbs state $\sigma_\beta\propto {\operatorname{exp}(-\beta H)}$ is a fixed point of the evolution. It is also unique, since $A^{a,\alpha}$ are a complete set of generators.
Finally notice that, in the limit $\beta \rightarrow 0$,
$\mathcal{L}^{(0)\dagger}(X):=\lambda \sum_{a\in \Lambda}\left( \frac{1}{2}\operatorname{tr}_a(X)-X\right)$,
where $\lambda=\frac{1}{\sqrt{2}e^{\frac{1}{4}}}$. This is the generator of the fully depolarizing channel.

\emph{Rapid mixing for high temperatures—} We consider local Hamiltonians $H$ on $\Lambda$ defined through a map from any non-empty finite set $X\subset \Lambda$ to a Hermitian operator $h_X$ supported in $X$ such that $\max_{X} \norm{h_X}_\infty \le h$. Given a subset of the lattice  $A \subset \Lambda$, we also define the Hamiltonian in that subset as
\begin{align*}
		H_A:=\sum_{X\subseteq A}h_X\,.
\end{align*}

We first assume that each $h_X$ has support on at most $k$ sites, and that each site $a$ appears on at most $l$ non-trivial $h_X$. The constants $h,k$ and $l$ are assumed to be independent of system size $n$, and we refer to those as $(k,l)$-local Hamiltonians. Let us define $J\equiv hkl$.

\begin{thm}\label{thmrmhighT1}
Let $H$ be $(k,l)$-local on a $D$ dimensional lattice, and let $\beta^*=\frac{1}{615^{D}2J}$. Then, for any $\beta<\beta^*$ and any initial state $\rho$,

\begin{align}\label{equ:mixing_gap}
\|e^{t\mathcal{L}^{(\beta)}}(\rho)-\sigma_\beta\|_1\le \epsilon\quad \text{ for all}\quad t=\Omega(\log(n/\epsilon))\,. 
\end{align}

\end{thm}

This statement of rapid mixing is our main result. We outline the proof steps below, with further details in \cite{supplemental}. 

The efficient preparation of $\sigma_\beta$ follows from rapid mixing together with~\cite[Theorem I.2]{chen2023efficient}, where it is shown that each term in the sum $\sum_{a,\alpha}$ of the generator \eqref{eq:generator1} can be implemented on a quantum computer with an effectively linear cost in $t$.
\begin{coro}
   Under the conditions of Theorem \ref{thmrmhighT1}, we can prepare an $\epsilon$-approximation of $\sigma_\beta$ on a quantum computer with $\widetilde{\mathcal{O}}(n)$ Hamiltonian simulation time and $\widetilde{\mathcal{O}}(n)$ two-qubit gates, where $\widetilde{\mathcal{O}}$ includes polylogarithmic terms in $\epsilon$ and $n$.
\end{coro}
We now provide the proof of Theorem \ref{thmrmhighT1} divided in four steps, with some of the technical details placed in \cite{supplemental}. The idea is to use an unconventional norm, the \emph{oscillator norm}, which allows us to extrapolate the rapid mixing property of the depolarizing channel to that of our Gibbs sampler, which can be seen as a perturbation. The strength of that perturbation increases with $\beta$, and is ultimately controlled by locality and the Lieb-Robinson bound. From those considerations the threshold $\beta^*$ follows.
\medskip
\paragraph{The oscillator norm as a proxy:} The trace norm $\norm{\cdot}_1$ is insensitive to local perturbations, and thus a direct analysis of it is inadequate for local-to-global proof ideas often used to prove rapid mixing. This motivates the definition of the so-called \emph{oscillator norm}
\begin{align*}
\vertiii{X}:=\sum_{a\in\Lambda}\|\delta_a(X)\|_\infty\,,
\end{align*}
where $\delta_a(X):=X-\frac{1}{2}I_a\otimes \operatorname{tr}_a(X)$. We aim to prove that, for $\beta< \beta^*$, there exists $\kappa\equiv \kappa(\beta)\in(0,\lambda)$ such that, for any initial observable $X$,
\begin{align}\label{gradientest}
\vertiii{e^{t\mathcal{L}^{(\beta)\dagger}}(X)}\le e^{-(\lambda-\kappa)t}\vertiii{X}\,.
\end{align}
Then, denoting $\rho_t:=e^{t\mathcal{L}^{(\beta)}}(\rho)$ and $X_t:=e^{t\mathcal{L}^{(\beta)\dagger}}(X)$,
\begin{align*}
\|\rho_t-\sigma\|_1&=\sup_{\|X\|_\infty\le 1}\tr{X(\rho_t-\sigma)}\\
&=\sup_{\|X\|_\infty\le 1}\|X_t-2^{-n}\tr{X_t}I\|\|\rho-\sigma\|_1\\
&\le \sup_{\|X\|_\infty\le 1} \vertiii{X_t}\,\|\rho-\sigma\|_1 \\
&\le e^{-(\lambda-\kappa)t}\sup_{\|X\|_\infty\le 1}\vertiii{X}\,\|\rho-\sigma\|_1\\
&\le 4n e^{-(\lambda-\kappa)t}\,,
\end{align*}
where in the second line we used the fact that $\rho_t - \sigma$ is traceless, and the last inequality uses the fact that $\vertiii{X}\le 2n$. Hence, rapid mixing is a consequence of the existence of a constant $0<\kappa<\lambda$ independent of $n$ that guarantees the decay of the oscillator norm. 

\medskip

\paragraph{Local conditions for oscillator norm decay:} The classical analogue of the oscillator norm \eqref{gradientest} has previously been considered in the classical literature on Markov chains. Its decay at high enough temperature can be derived by considering its differential form and identifying a set of sufficient conditions on the local generators of the evolution \cite{Aizenman1987,stroock1992equivalence}. This idea was extended to the commuting Hamiltonian setting in \cite{Majewski1995,temme2015fast}, and here we generalize it to the noncommuting regime. In \cite{supplemental} we show that \eqref{gradientest} holds if the constant $\kappa$ is such that
\begin{align*}
\sum_a\sum_{b\ne a}\kappa_{a,b}^c+\sum_a \gamma_a^c\le \kappa 
\end{align*}
for positive constants $\kappa_{a,b}^c\ge 0$ and $\gamma^c_a\ge 0$ depending on lattice sites $a,b,c$ where, for all $X$ and $b\ne a$, $\kappa_{a,b}^c$ and $\gamma_a^c$ are such that
\begin{align}
&\|[\delta_a, \mathcal{L}_b^{(\beta)\dagger}](X)\|_\infty\le \sum_c \kappa^c_{a,b} \|\delta_c(X)\|_\infty\label{eq:11}\\
&\|\delta_a(\mathcal{L}_a^{(\beta)\dagger}-\mathcal{L}_a^{(0)\dagger})(X)\|_\infty\le \sum_c \gamma_a^c\|\delta_c(X)\|_\infty\,.\label{eq:22}
\end{align}

These conditions are reminiscent of Dobrushin's uniqueness condition in the classical setting \cite{dobrushin1970prescribing}. In \cite{supplemental}, we show that our criteria for rapid mixing are in fact weaker than that condition when restricting to the classical setting.

\paragraph{Local conditions via Lindblad approximations:} In order to analyze the constants  $\kappa_{a,b}^c$ and $\gamma_a^c$, we need to approximate the Lindbladians in space and in $\beta$ as follows. Given a constant $r_0\in\mathbb{N}$ to be chosen later,
\begin{align*}
\sum_{r\ge r_0}\|\mathcal{L}_b^{\beta,r\dagger}\!\!-\mathcal{L}_{b}^{\beta,r-1\dagger }\|_{\infty\to\infty}\!&\le \Delta(r_0)\,,\\
\|\mathcal{L}_b^{(\beta) \dagger}-\mathcal{L}_{b}^{(0) \dagger }\|_{\infty\to\infty}&\le \eta(\beta)\,,
\end{align*}
where $\mathcal{L}_b^{\beta,r\dagger}$ corresponds to the local Lindbladian $\mathcal{L}_b^{(\beta)}$ with associated Hamiltonian $H$ replaced by $H_{B_b(r)}$, i.e.~the sum of interactions supported in a ball of radius $r$ aroud site $b$. Above,
 $\Delta$ is a fast decaying function, whereas $\eta$ is an increasing function of $\beta$ with $\eta(0)=0$. 
 
 Then, in \cite{supplemental}, we show that $\kappa_{a,b}^c$ and $\gamma_a^c$ can be chosen so that, for any $r_0\in\mathbb{N}$ and $c\in\Lambda$, 
 \begin{align}
&\sum_a\sum_{b\ne a}\kappa_{a,b}^c+\sum_a \gamma_a^c \le 4(2r_0+1)^{2D}\,\eta(\beta)\label{boundcrude}+f(r_0),
\end{align}
where $f(r_0)$ is a fast-decaying function depending on $\Delta(l)$ and $D$, whose expression we give in \cite{supplemental}.
Hence, in order to obtain $\kappa < \lambda$, it is enough that the RHS is also smaller than $\lambda$. This means we need to take $r_0$ large enough and choose $\beta$ small enough accordingly.

\paragraph{Spatial and thermal approximations via Lieb-Robinson bounds:} In \cite{supplemental}, we derive the following expressions. Assuming initially that $\beta J\le 1/200$,
\begin{align*}
&\Delta(\ell):= \frac{14 g(\beta J)^{r_0}}{1-g(\beta J)}\,,\\
  &  \eta(\beta) := 2 \beta J \,,
\end{align*}
where $g(x)=\sqrt{x}/(1+\sqrt{x})$.
Plugging these into the bounds derived in \eqref{boundcrude}, $\kappa$ may be chosen as
\begin{align*}
\kappa& \equiv 8(2 r_0+1)^{2D} \beta J +f(r_0),
\end{align*}
and if we choose $r_0=4$ and $2 \beta J\le 1/615^D$, we find that $\kappa$ is strictly less than $\lambda$. For low-dimensional models with large connectivity $k$, this includes lower temperatures than the results of \cite{tang_efficient} (where $\beta^*$ depends on $k^{-2}$).

\emph{Long-range Hamiltonians—} The result above holds as long as the Hamiltonian obeys a strong enough Lieb-Robinson bound. As such, it also extends to long range models, for which these bounds have been proven \cite{Kuwahara2019LongRange,TranLR2021}. Long-range Hamiltonians are such that
\begin{align}
 &  \sup_{i \in \Lambda} \sum_{\substack{X \in \Lambda, i \in X \\ \text{diam}(X)\le r}} \norm{h_{X}}_{\infty} \le g r^{D-\nu}, \\
 &   \sup_{i,j \in \Lambda} \sum_{\{i,j \} \in X} \norm{h_{X}}_{\infty} \le g_0 (\text{dist}(i,j)+1)^{-\nu},
\end{align}
for some constants $g,g_0$. The parameter $\nu$ controls the range of the interactions.

\begin{thm}\label{thm.longrange}
Let $H$ be long range  and such that $\nu>4D+2$. There exists a constant $C$ 
depending solely on $D, g_0$ and $\nu$ such that for $\beta < \beta^*:= C/g $,
\begin{align}\label{equ:mixing_gap}
\|e^{t\mathcal{L}^{(\beta)}}(\rho)-\sigma_\beta\|_1\le \epsilon\quad \text{ for all}\quad t=\Omega(\log(n/\epsilon))\,.
\end{align}

\end{thm}

The proof steps $a-c$ are the same as for the strictly local case, but now in \eqref{boundcrude} in step $d$ we have the functions
\begin{align*}
&\Delta(\ell):= K  (\beta g)^{\frac{D+1}{2}} \frac{1}{r_0^{\nu -2D-2 }} , \\ 
  &  \eta(\beta) := 2 \beta g \,,
\end{align*}
where the constant $K$ depends on the parameters of the Lieb-Robinson bound \cite{Kuwahara2019LongRange} (see \cite{supplemental} for details). The condition $\nu > 4D+2 $ then guarantees that the sum inside $f(r_0)$ in \eqref{boundcrude} is finite and decays with $r_0$, which allows us to pick a numerical constant $C$ such that $\kappa<\lambda$ as required. Notice that the range of interactions allowed is narrower than the Lieb-Robinson bound, which holds as long as $\nu > 2D$ \cite{LR_Anthony_Chen_2023}. This is because of the overheads in the definition of $\kappa$, which contain a triple sum of the error term of the Lieb-Robinson bound. 

\emph{Comparison with previous work—} Until recently, results for quantum algorithms preparing Gibbs states of general Hamiltonians only guaranteed convergence in exponential time \cite{poulin_sampling_2009,Temme_2011,chowdhury2016quantum,shtanko2021algorithms,Holmes_2022,zhang2023dissipative}. Some approaches focused on modelling aspects of the thermalization process \cite{chen2023fast,LiWang2023,Rall_2023}, and others resorted to variational approaches \cite{chowdhury_variational_2020,wang_variational_2021} or the notion of quantum Markov states \cite{brandao2019finite}. Polynomial-time algorithms were constrained to 1D systems \cite{bilgin_preparing_2010} or commuting Hamiltonians \cite{GeMolnar2016,kastoryano2016quantum,bardet2023rapid,bardet2024entropy}, for which there has also been further recent progress \cite{ding2024polynomial,hwang2024gibbscommuting}. Also, \cite{brandao2019finite} gave a scheme based on the Petz map with a superpolynomial scaling in $D>1$ under the assumption of decay of correlations and a fast decay of the quantum conditional mutual information \cite{kuwahara2024clustering}.

To our knowledge, Theorems \ref{thmrmhighT1} and \ref{thm.longrange}, together with the algorithm from \cite{chen2023efficient}, represent the first method that enables quasi-linear preparation for general non-commuting models, albeit restricted to high temperatures. Previously, the authors had derived quasi-quadratic bounds on the runtime of the same Gibbs samplers by means of spectral gap estimates
 in \cite{rouze2024efficient}, which is a strictly weaker property than rapid mixing.

In a recent work \cite{tang_efficient}, the authors derived an efficient Gibbs sampling algorithm for short-range systems over arbitrary hypergraphs at high temperature (roughly with $\beta^* =\frac{1}{100J}$ in our notation). However their algorithm calls a classical subroutine with runtime scaling as $n^{7+\mathcal{O}(\log(D))}$ for $D$-dimensional lattice systems. Additionally, the temperature range within which our dissipative preparation is efficient is larger than theirs for low-dimensional, highly connected systems. 
 
 Crucially, the result of \cite{tang_efficient} makes use of the fact that high-temperature Gibbs states are separable, a phenomenon previously anticipated and known as ``sudden death of entanglement'' \cite{Sherman_2016}. It remains an intriguing  problem to show a separation between the temperature at which the Gibbs state becomes unentangled and that at which dissipative preparation becomes inefficient. Another open question is whether that sudden death of entanglement occurs for long-range Hamiltonians. 
 
 


\emph{Approximating quantum partition functions—}
We now discuss the application of our Gibbs sampling results to estimate partition functions of high-temperature quantum Gibbs states to an $\epsilon$ multiplicative error, where the estimate $Z'_\beta$ is such that
$\vert Z'_\beta - Z_\beta \vert \le  \epsilon Z_\beta$.

A standard class of classical algorithms for partition function estimation relies on an annealing schedule of inverse temperatures, together with the ability to sample from Gibbs states to estimate the partition function via telescopic products \cite{Stefankovic2007,pmlr-v75-kolmogorov18a,Huber2015,Bezkov2008}. Previous results~\cite{Wocjan2008,poulin_sampling_2009,Montanaro2015,hamoudi_magniez,Harrow2020,Arunachalam2022,Cornelissen2023} generalized this approach to the quantum setting, and here we give a more modern exposition of the algorithm. Assuming access to a block encoding of $\frac{H}{\|H\|_\infty}$ and Hamiltonian simulation time for $H$, we show how to leverage our Gibbs preparation results to estimate partition functions.
The procedure is explained in detail in \cite{supplemental}, but roughly works as follows: we use our efficient Gibbs sampler to prepare a sequence of Gibbs states $\sigma_{\beta_1},\ldots,\sigma_{\beta_l}$, such that $\beta_1=0$ and $\beta_l=\beta$ is our target inverse temperature. Given these states, we then use quantum signal processing techniques and the block encoding of $H$ to obtain an approximate block encoding of $e^{-\tfrac{1}{2}(\beta_{i+1}-\beta_i)H}$ on $\sigma_{\beta_i}$. Applying this to the Gibbs state $\sigma_{\beta_i}$ with an auxilliary qubit, we are able to sample from a random variable with expectation value $Z_{\beta_{i+1}}/Z_{\beta_{i}}$ and bounded variance.
By picking the $|\beta_i-\beta_{i+1}|=\Theta(n^{-1})$ and drawing $\mathcal{O}(n\epsilon^{-2})$ samples for each temperature pair, we can estimate $Z_\beta$ to precision $\epsilon$ with probability of success at least $3/4$. This is simply done by taking the product of the estimates of $Z_{\beta_{i+1}}/Z_{\beta_{i}}$, and applying the results of~\cite{dyer1991computing} to bound the error. 
\begin{thm}
For a short-range lattice Hamiltonian $H$ and $\beta\leq\beta^*$ as in Theorem~\ref{thmrmhighT1} we can estimate the partition function $Z_{\beta}:=\tr{e^{-\beta H}}$ up to a relative error $\epsilon>0$ with success probability at least $3/4$ and runtime
\begin{align}
\widetilde{\mathcal{O}}(n^3\epsilon^{-2}),
\end{align}
where $\widetilde{\mathcal{O}}$ hides polylogarithmic terms in $\epsilon$, and $n$. For long range Hamiltonians as in Theorem~\ref{thm.longrange} we can solve the same task with runtime $\widetilde{\mathcal{O}}(n^4\epsilon^{-2})$.
\end{thm}

We believe that it should be possible to remove another $n$ factor from the complexity of our algorithm by considering adaptive annealing schedules.
The reason for that comes from the literature on the topic of designing annealing schedules for \emph{classical systems}~\cite{Arunachalam2022,Stefankovic2007,Harrow2020,Huber2015,pmlr-v75-kolmogorov18a,Montanaro2015,Bezkov2008,hamoudi_magniez,Cornelissen2023}. In these works, the authors give algorithms to obtain good adaptive annealing schedules with $l=\mathcal{O}(\sqrt{n})$, which would yield the speedup of order $n$ in our result. However, to the best of our knowledge, most works assume that $H$ is classical and its spectrum only takes integer values. As remarked in~\cite[Footnote 1]{Arunachalam2022}, ``this assumption is common in the literature and can often be relaxed''. We leave it for future work to generalize these results to obtain annealing schedules for quantum Hamiltonians with arbitrary spectrum. Furthermore, to close the gap between local and long-range models, it would be interesting to develop more tailored simulation results for the Lindbladians related to long-range Hamiltonians. Although the Hamiltonian simulation results we use from~\cite{Haah_2021} were extended to the long range setting in~\cite{TranLong19}, the polynomial scaling in precision of the results makes them less efficient than block-encoding methods for long-range systems for our application.

Note that for the algorithm  above, the runtime is at most cubic in system size and the dependency on the dimension of the lattice and other parameters such as temperature appears only as a prefactor.
This should be compared to other (classical) methods to compute the partition function~\cite{Mann2021,Harrow2020b,Haah_2024}. These achieve polynomial runtimes in $n,\epsilon^{-1}$, but have a higher polynomial dependency  for the same task in the high-temperature regime. In particular, the polynomial degree of $n,\epsilon^{-1}$ in the runtime of those algorithms usually depends on the dimension of the underlying lattice or on the temperature. Thus, we see that our methods give a polynomial quantum speedup for this task, although in a more restricted temperature range. The only regime where we achieve a super-polynomial speed-up over classical methods is for systems with long-range interactions, for which the best known rigorous runtimes at high temperatures are only $\text{exp}({\log^2(n/\epsilon)})$\cite{sanchezsegovia2025hightemperature}.

Several papers also studied quantum speedups for partition functions assuming access to quantum samples of the classical Gibbs state~\cite{Montanaro2015,hamoudi_magniez,Arunachalam2022,Cornelissen2023}, i.e. a state of the form 
\begin{align}
\ket{\psi}=\sum\limits_{x\in\{0,1\}^n}\sqrt{\frac{e^{-\beta H(x)}}{Z_\beta}}\ket{x},
\end{align}
or to a unitary preparing it, where $H$ is a classical Hamiltonian. Under this assumption, these works construct quantum algorithms that obtain a Heisenberg scaling for estimating the partition function, i.e. the complexity in precision scales like $\epsilon^{-1}$ compared to the usual $\epsilon^{-2}$ of classical algorithms.
In principle, it should be possible to adapt the quantum algorithms of~\cite{Arunachalam2022}, which are designed to work with classical Hamiltonians, to work given access to a purified Gibbs state, sometimes also called a thermofield double. However, our techniques currently do not yield preparation procedures for the purified Gibbs with sublinear scaling with $\epsilon$, defeating the purpose of obtaining a Heisenberg scaling. 

\emph{Conclusion—}
We established the first rapid mixing results for quantum Gibbs states of non-commuting lattice Hamiltonians at high temperature, yielding essentially optimal mixing times, and with explicit constants for the threshold temperature. Furthermore, our results are the first to also apply to long range systems. We then leveraged these findings to obtain efficient quantum algorithms to estimate the partition function of these Gibbs states. For the case of short range Hamiltonians, these give large polynomial speedups when compared with existing classical methods, and for long range systems they give the first provably efficient schemes. As such, our results substantially advance our understanding on the complexity of simulating high temperature quantum Gibbs states. Finally, we expect that our methods also extend to further families of KMS Lindbladians, such as \cite{ding2024efficientquantumgibbssamplers,scandi2025thermalizationopenmanybodysystems,ding2025endtoendefficientquantumthermal}.

\emph{Acknowledgements—} 
CR acknowledges financial support from the ANR project QTraj (ANR-20-CE40-0024-01) of the French National Research Agency (ANR). DSF acknowledges funding from the European Union under Grant Agreement 101080142 and the project EQUALITY and from the Novo Nordisk
Foundation (Grant No. NNF20OC0059939 Quantum for Life). AMA acknowledges support from the Spanish Agencia Estatal de Investigacion through the grants ``IFT Centro de Excelencia Severo Ochoa CEX2020-001007-S" and ``Ram\'on y Cajal RyC2021-031610-I'', financed by MCIN/AEI/10.13039/501100011033 and the European Union NextGenerationEU/PRTR. This
project was funded within the QuantERA II Programme that has received funding from the EU’s H2020 research and innovation programme under the GA No 101017733.

\bibliographystyle{apsrev4-2}
\bibliography{references}

\appendix

\pagestyle{fancy}
\fancyhf{} 
\fancyfoot[C]{\thepage} 
\renewcommand{\headrulewidth}{0pt} 
\renewcommand{\footrulewidth}{0pt}
\widetext

\section{Relation of oscillator norm bounds to Dobrushin's classical condition}

	\noindent The conditions we require for the Lindbladian to be rapidly mixing in the main text are as follows
	\begin{align}
		&\|[\delta_a, \mathcal{L}_b^{(\beta)\dagger}](X)\|_\infty\le \sum_c \kappa^c_{a,b} \|\delta_c(X)\|_\infty\label{eq:11}\\
		&\|\delta_a(\mathcal{L}_a^{(\beta)\dagger}-\mathcal{L}_a^{(0)\dagger})(X)\|_\infty\le \sum_c \gamma_a^c\|\delta_c(X)\|_\infty\,.\label{eq:22}
\end{align}

	These are reminiscent of Dobrushin's uniqueness condition in the classical setting, as we now explain. Let us describe the dynamics towards a classical Gibbs distribution $p_\beta$ with the \emph{heat-bath generators}. For any region $A\subset \Lambda$ and any boundary condition $\tau\in \{-1,1\}^{\Lambda\backslash A}$, consider the conditional expectations $\mathbb{E}_{A}^{\tau}:f\mapsto \sum_{\eta_A}f(\eta_A,\tau)p^{\tau}_\beta(\eta_A)$, where $p_\beta^{\tau}$ denotes the conditional distribution over spin configurations in $A$ conditioned on configuration $\tau$ outside. Then, for any $\eta=(\eta_1,\dots,\eta_n)\in\{-1,1\}^{\Lambda}$ with restrictions $\eta^j=(\eta_1,\dots,\eta_{j-1},\eta_{j+1}\dots ,\eta_n)$, the heat-bath generators are
	\begin{align*}
		L^{(\beta)}_A(f)(\eta)=\sum_{j\in A} \mathbb{E}_{\{j\}}^{\eta^j}(f)(\eta^j)-f(\eta)\,.
\end{align*}

	In this case, Eq. \eqref{eq:22} is trivially satisfied, since the difference $L_a^{(\beta)}-L_a^{(0)}$ maps to functions independent on the $\eta_a$. Moreover, the LHS of Eq. \eqref{eq:11} reduces to
	$\|[\mathbb{E}^{{u}}_{\{a\}},\mathbb{E}_{\{b\}}] \|$, where $\mathbb{E}^{{u}}_{\{a\}}$ denotes the local uniform average over spin configurations on site $a$. Further simplifying this commutator, we see that
	\begin{align*}
		&[\mathbb{E}^{{u}}_{\{a\}},\mathbb{E}_{\{b\}}](f)(\eta)\\
		&\quad= \frac{1}{2}\sum_{\sigma_a,\sigma_b}f(\sigma_a,\eta^{\{a,b\}},\sigma_b)\Big( p_\beta^{\eta^{\{a,b\}},\sigma_a}(\sigma_b)-p_\beta^{\eta^{b}}(\sigma_b)\Big)\\
		&\quad =\frac{1}{2}\sum_{\sigma_b}f(\sigma_a,\eta^{\{a,b\}},\sigma_b)\big(p_\beta^{\eta_-^{\{a,b\}}}(\sigma_b)-p_\beta^{\eta_+^{\{a,b\}}}(\sigma_b)\big)\\
		&\quad =\frac{1}{2}\sum_{\sigma_b}\big(f(\sigma_a,\eta^{\{a,b\}},\sigma_b)-\mathbb{E}_b[f(\sigma_a,\eta^{\{a,b\}},.)]\big)\\
		&\qquad\qquad\qquad\qquad\qquad\qquad \times\big(p_\beta^{\eta_-^{\{a,b\}}}(\sigma_b)-p_\beta^{\eta_+^{\{a,b\}}}(\sigma_b)\big)
	\end{align*}
	where $\eta^{\{a,b\}}$ denotes the vector of spin configurations outside $\{a,b\}$ and $\eta^{\{a,b\}}_{\pm}$ that for which configuration in $a$ is fixed to $\pm$. 
	The supremum over all bounded functions is attained at the total variation
	\begin{align*}
		c_{b,a}:=\!\!\max_{{\eta^{\{a,b\}}\in\{-1,1\}^{n-2}}}\Big\|p_\beta^{\eta_+^{\{a,b\}}}-p_\beta^{\eta_-^{\{a,b\}}}\Big\|_{\operatorname{TV}}
	\end{align*}
	which corresponds to Dobrushin's influence matrix \cite{dobrushin1970prescribing}. Then
	\begin{align*}
		\|[\mathbb{E}^{{u}}_{\{a\}},\mathbb{E}_{\{b\}}](f)\|\le  c_{b,a}\|f-\mathbb{E}_b(f)\|\,.
	\end{align*}
	Dobrushin's uniqueness criterion then refers to the condition that $\|c\|_1:=\sup_b\sum_{a\ne b}c_{b,a}<1$. This directly implies that the tensor $\kappa_{a,b}^c=\mathbf{1}_{b=c}c_{c,a}$ satisfies
	\begin{align*}
		\sup_c\sum_a\sum_{b\ne a}\kappa^c_{a,b}=\sup_c\sum_{a\ne c}c_{c,a} =\|c\|_1<1\,.
	\end{align*}
	This shows that our criteria for rapid mixing are weaker than Dobrushin's uniqueness condition when restricting to the classical setting.

\section{Rapid mixing at high temperature}\label{sec:highTgibbsrm}

\subsection{Local conditions for oscillator norm decay}

First, let us define each of the local terms of the generator
\begin{align}\label{eq:generator}
	\mathcal{L}^{(\beta)}(\rho)&
	\equiv \sum_{a\in\Lambda}\mathcal{L}^{(\beta)}_a(\rho)\,,
\end{align}
corresponding to the contribution of the jump operators at a fixed site $a$, $\{A^{a,\alpha}\}_{\alpha=0}^3$. We consider the non-commutative quasi-derivations $\delta_a(X):=X-\frac{1}{2}\,I_a\otimes \operatorname{tr}_a(X)$. We have that, denoting $X_s:=e^{s\mathcal{L}^{(\beta)\dagger}}(X)$
\begin{align}
	\frac{d}{ds}\delta_a (X_s)&=\delta_a\mathcal{L}^{(\beta)\dagger}(X_s)\\
	&=\delta_a \mathcal{L}^{(0)\dagger}_a(X_s)+\delta_a(\mathcal{L}^{(\beta)\dagger}_a-\mathcal{L}^{(0)\dagger}_a)(X_s)+\sum_{b\ne a}\delta_a\mathcal{L}_b^{(\beta)\dagger}(X_s)\,.
\end{align}

In Appendix B.2 of \cite{rouze2024efficient} it was shown that $\mathcal{L}_a^{(0)\dagger}(X):=\lambda\left( \frac{1}{2}\operatorname{tr}_a(X)-X\right)$ is the generator of the depolarizing semigroup on qubit $a$ with $\lambda=\frac{1}{\sqrt{2}e^{\frac{1}{4}}}$. We thus get that $\delta_a\mathcal{L}_a^{(0)}=-\lambda \delta_a$ and hence
\begin{align}
	\frac{d}{ds}\delta_a (X_s)=-\lambda\delta_a(X_s)+\delta_a(\mathcal{L}^{(\beta)\dagger}_a-\mathcal{L}^{(0)\dagger}_a)(X_s)+\sum_{b\ne a}\mathcal{L}_b^{(\beta)\dagger}\delta_a(X_s)+\sum_{b\ne a}[\delta_a, \mathcal{L}_b^{(\beta)\dagger}](X_s)\,.
\end{align}
Denoting by $\mathcal{P}_t$ the semigroup generated by $\mathcal{L}^{(\beta)\dagger}$, and by $\mathcal{P}_t^{(a)}$ the one generated by $\sum_{b\ne a}\mathcal{L}_b^{(\beta)\dagger}$, we have that for any $t\ge s$ 
\begin{align}
	\frac{d}{ds}\left(e^{\lambda s}\mathcal{P}_{t-s}^{(a)}\delta_a \mathcal{P}_s(X)\right)=e^{\lambda s}\left(\sum_{b\ne a}\mathcal{P}_{t-s}^{(a)}[\delta_a, \mathcal{L}_b^{(\beta)\dagger}](X_s)+\mathcal{P}_{t-s}^{(a)}\delta_a(\mathcal{L}^{(\beta)\dagger}_a-\mathcal{L}^{(0)\dagger}_a)(X_s)\right)\,.
\end{align}
Integrating the above from $0$ to $t$ and using the contractivity of the semigroups, we have that
\begin{align}
	\|\delta_a \mathcal{P}_t(X)\|_\infty&\le e^{-\lambda t}\|\mathcal{P}_{t}^{(a)}\delta_a (X)\|_\infty\\
	&~~~~~~~~ + \int_{0}^t e^{\lambda (s-t)}\left(\sum_{b\ne a}\|\mathcal{P}_{t-s}^{(k)}[\delta_a, \mathcal{L}_b^{(\beta)\dagger}](X_s)\|_\infty+\|\mathcal{P}_{t-s}^{(a)}\delta_a(\mathcal{L}^{(\beta)\dagger}_a-\mathcal{L}^{(0)\dagger}_a)(X_s)\|_\infty\right)\,ds\\
	&\le e^{-\lambda t} \|\delta_a (X)\|_\infty+ \int_{0}^t e^{\lambda (s-t)}\left(\sum_{b\ne a}\|[\delta_a, \mathcal{L}_b^{(\beta)\dagger}](X_s)\|_\infty+\|\delta_a(\mathcal{L}^{(\beta)\dagger}_a-\mathcal{L}^{(0)\dagger}_a)(X_s)\|_\infty\right)\,ds
\end{align}
If we prove the existence of positive constants $\kappa_{a,b}^c\ge 0$ and $\gamma^c_a\ge 0$ such that, for all $X$ and $b\ne a$,
\begin{align} \label{eq:condition}
	&\|[\delta_a, \mathcal{L}_b^{(\beta)\dagger}](X)\|_\infty\le \sum_c \kappa^c_{a,b} \|\delta_c(X)\|_\infty\\
	&\|\delta_a(\mathcal{L}_a^{(\beta)\dagger}-\mathcal{L}_a^{(0)\dagger})(X)\|_\infty\le \sum_c \gamma_a^c\|\delta_c(X)\|_\infty \nonumber
\end{align}
with $\sum_a\sum_{b\ne a}\kappa_{a,b}^c+\sum_a \gamma_a^c\le \kappa $, we would get that, denoting the oscillator norm as $\vertiii{.}:=\sum_a \|\delta_a(.)\|_\infty$,
\begin{align}
	\vertiii{X_t}\le e^{-\lambda t}\vertiii{X}+\kappa\int_0^t e^{\lambda(s-t)}\,  \vertiii{X_s}\,ds\,.
\end{align}
From here, the conclusions of \cite{Majewski1995,temme2015fast} follow. That is, as long as $\kappa<\lambda$, for all $X$,
\begin{align}
	\vertiii{X_t}\le e^{-(\lambda-\kappa)t}\vertiii{X}\,.
\end{align}

\subsection{Local conditions via Lindblad approximations}

We now show how the constants defined in Eq.~\eqref{eq:condition} above satisfy $\sum_a\sum_{b\ne a}\kappa_{a,b}^c+\sum_a \gamma_a^c\le \kappa $ for suitably well behaved Lindbladians. For this, we show in Appendix \ref{sec:local} below that the Lindbladians can be approximated in space and in $\beta$ as follows:
\begin{align}
	&\|\mathcal{L}_b^{\beta,r\dagger}-\mathcal{L}_{b}^{\beta,r-1\dagger }\|_{\infty\to\infty}\le \zeta(r),\\
	&\|\mathcal{L}_b^{(\beta) \dagger}-\mathcal{L}_{b}^{(0) \dagger }\|_{\infty\to\infty}\le \eta(\beta)\,,
\end{align}
where $\mathcal{L}_b^{\beta,r\dagger}$ corresponds to the local Lindbladian $\mathcal{L}_b^{(\beta)\dagger}$ with associated Hamiltonian $H$ replaced by $H_{B_b(r)}$, i.e.~the sum of interactions supported in a ball of radius $r$ aroud site $b$. Above,
$\zeta$ is a fast decaying function with $\sum_{r\ge r_0}\zeta(r)\equiv \Delta(r_0)<\infty$, whereas $\eta(\beta)$ is an increasing functions of $\beta$ such that $\eta(0)=0$. Then, we have that 
\begin{align}
	\|[\delta_a, \mathcal{L}_b^{(\beta)\dagger}](X)\|_\infty&= \|[\delta_a, (\mathcal{L}_b^{(\beta)\dagger}-\mathcal{L}_b^{\beta,\operatorname{dist}(a,b)\dagger})](X)\|_\infty\\
	&\le \|\mathcal{L}_b^{(\beta)\dagger}-\mathcal{L}_b^{\beta,\operatorname{dist}(a,b)\dagger}\|_{\infty\to\infty}\|\delta_a(X)\|_\infty+2\|(\mathcal{L}_b^{(\beta)\dagger}-\mathcal{L}_b^{\beta,\operatorname{dist}(a,b)\dagger})(X)\|_\infty\\
	&\le \sum_{r>\operatorname{dist}(a,b)} \zeta(r)\|\delta_a(X)\|_\infty+  2 \sum_{r>\operatorname{dist}(a,b)} \|(\mathcal{L}_b^{\beta,r\dagger}-\mathcal{L}_b^{\beta,r-1\dagger})\delta_{B_b(r)}(X)\|_\infty\\
	&\le \Delta(\operatorname{dist}(a,b))\|\delta_a(X)\|_\infty+  2 \sum_{r>\operatorname{dist}(a,b)} \zeta(r)\,\sum_{\operatorname{dist}(b,c)\le r}\|\delta_{c}(X)\|_\infty\\
	&=\Delta(\operatorname{dist}(a,b))\|\delta_a(X)\|_\infty+  2 \sum_c \|\delta_{c}(X)\|_\infty\Delta(\max(\operatorname{dist}(a,b),\operatorname{dist}(b,c)))\,.
\end{align}
For a parameter $r_0$ to be chosen later and $\operatorname{dist}(a,b)>r_0$, we choose $\kappa_{a,b}^a=\Delta(\operatorname{dist}(a,b))$ and $\kappa_{a,b}^c=2 \Delta(\max(\operatorname{dist}(a,b),\operatorname{dist}(b,c)))$ for any other site $c$ with $c\ne a$. In contrast, whenever $\operatorname{dist}(a,b)\le r_0$, we control the above commutator as follows
\begin{align}
	&\|[\delta_a, \mathcal{L}_b^{(\beta)\dagger}](X)\|_\infty\le \|[\delta_a, (\mathcal{L}_b^{(\beta)\dagger}-\mathcal{L}_b^{\beta,r_0\dagger})](X)\|_\infty+\|[\delta_a, (\mathcal{L}_b^{\beta,r_0\dagger}-\mathcal{L}_b^{0,r_0\dagger})](X)\|_\infty\\
	&\qquad\le (\eta(\beta)+\Delta(r_0))\|\delta_a(X)\|_\infty+ 2\sum_{r>r_0}\zeta(r)\!\!\!\!\sum_{\operatorname{dist}(c,b)\le r}\|\delta_c(X)\|_\infty+2\eta(\beta)\!\!\!\!\!\sum_{\operatorname{dist}(b,c)\le r_0}\!\!\!\!\!\|\delta_c(X)\|_\infty\\
	&\qquad\le (\eta(\beta)+\Delta(r_0))\|\delta_a(X)\|_\infty+ 2\sum_c\Delta(\max(r_0,\operatorname{dist}(c,b)))\|\delta_c(X)\|_\infty+2\eta(\beta)\!\!\!\!\!\sum_{\operatorname{dist}(b,c)\le r_0}\!\!\!\!\!\|\delta_c(X)\|_\infty\,.
\end{align}
Hence, for $\operatorname{dist}(a,b)\le r_0$, we choose $\kappa_{a,b}^c=3\eta(\beta)+3\Delta(r_0)$ for any site $c$ with $\operatorname{dist}(c,b)\le r_0$, and $\kappa_{a,b}^c=2\Delta(\operatorname{dist}(b,c))$ otherwise. Similarly, we control
\begin{align}
	\|\delta_a(\mathcal{L}_a^{(\beta)\dagger}-\mathcal{L}_a^{(0)\dagger})(X)\|_\infty&\le \eta(\beta)\sum_{\operatorname{dist}(a,c)\le r_0}\|\delta_c(X)\|_\infty+2\sum_{r\ge r_0}\zeta(r)\sum_{\operatorname{dist}(a,c)\le r}\|\delta_c(X)\|_\infty\\
	&\le \eta(\beta)\sum_{\operatorname{dist}(a,c)\le r_0}\|\delta_c(X)\|_\infty+2\sum_c\Delta(\max(r_0,\operatorname{dist}(a,c))){\|\delta_c(X)\|_\infty}\,,
\end{align}
so that $\gamma_a^c=\eta(\beta)+2 \Delta(r_0)$ if $\operatorname{dist}(a,c)\le r_0$, and $\gamma_a^c=2\Delta(\operatorname{dist}(a,c))$ otherwise. Then, for any site $c$,
\begin{align}
	\sum_a\sum_{b\ne a}\kappa_{a,b}^c+\sum_a \gamma_a^c &= \sum_{0<\operatorname{dist}(a,b)\le r_0}\kappa_{a,b}^c+\sum_{\operatorname{dist}(a,b)>r_0}\kappa^c_{a,b}+\sum_a\gamma_a^c\\
	&\le 3(\eta(\beta)+\Delta(r_0))|\{(a,b)|\,0<\operatorname{dist}(a,b)\le r_0,\,\operatorname{dist}(b,c)\le r_0\}|\\
	&+\sum_{\operatorname{dist}(b,c)>r_0 }(2|B_b(r_0)|+1)\,\Delta(\operatorname{dist}(b,c))\\
	&+2\sum_{\operatorname{dist} \label{eq:a27}(a,b)>r_0}\Delta(\max(\operatorname{dist}(a,b),\operatorname{dist}(b,c)))\\
	&+|\{a:\operatorname{dist}(a,c)\le r_0\}|\Big(\eta(\beta)+2\Delta(r_0)\Big)+2\sum_{a:\operatorname{dist}(a,c)>r_0}\Delta(\operatorname{dist}(a,c))\,.
\end{align}
Considering that $|B_b(r_0)| \le (2r_0+1)^D$, notice that for the term in Eq. \eqref{eq:a27} we have
\begin{align}
	\sum_{\operatorname{dist}(a,b)>r_0}\Delta(\max(\operatorname{dist}(a,b),\operatorname{dist}(b,c))) &= \sum_{\operatorname{dist}(b,c)>\operatorname{dist}(a,b)>r_0}\Delta(\operatorname{dist}(b,c)) +\sum_{\substack{\operatorname{dist}(a,b)>r_0\\ \operatorname{dist}(b,c) \le \operatorname{dist}(a,b)}}\Delta(\operatorname{dist}(a,b)) 
	\\& \le 
	\sum_{b,\operatorname{dist}(b,c)>r_0} \vert \{a \vert \operatorname{dist}(a,b)\le \operatorname{dist}(b,c)\}\vert\Delta(\operatorname{dist}(b,c))
	\\ & \quad \quad +
	\sum_{b} \sum_{\substack{l \ge r_0 \\ l \ge \operatorname{dist}(b,c) }} \vert \{a \vert \operatorname{dist}(a,b)=l\}\vert \Delta(l)
	\\& \le
	\sum_{l\ge r_0} (2l+1)^{2D-1} \Delta(l) +\sum_{b} \sum_{\substack{l \ge r_0 \\ l \ge \operatorname{dist}(b,c) }} (2l+1)^{D-1} \Delta(l)
	\\& \le
	\sum_{l\ge r_0} (2l+1)^{2D-1} \Delta(l) +\sum_{l'} (2l'+1)^{D-1}\sum_{\substack{l \ge r_0 \\ l \ge l' }} (2l+1)^{D-1} \Delta(l)
	\\& \le
	\sum_{l\ge r_0} (2l+1)^{2D-1} \Delta (l) + r_0 
	\sum_{l \ge r_0} (2l+1)^{2D-2} \Delta (l) + \sum_{l' \ge r_0} \sum_{l \ge l'} (2l+1)^{2D-2} \Delta(l).
\end{align}
By simplifying these expressions we can bound, in terms of $\eta(\beta)$, $\Delta(l)$ and $r_0$, 
\begin{align}
	\sum_a\sum_{b\ne a}\kappa_{a,b}^c+\sum_a \gamma_a^c & \le 4(2r_0+1)^{2D} \eta(\beta) + 5  (2r_0+1)^{2D} \Delta(r_0) \\&+\left( 5+2r_0+2(2r_0+1)^{D} \right)\sum_{l \ge r_0} (2l+1)^{2D-1}\Delta(l) + 2 \sum_{l'\ge r_0} \sum_{l \ge l'} (2l+1)^{2D-2} \Delta(l).\nonumber
	\\ & \equiv 4(2r_0+1)^{2D} \eta(\beta) + f(r_0).
\end{align}

Hence, for $\kappa$ to be suitably small we need to show that this expression is bounded away from $\lambda$. This requires explicit expressions for $\eta(\beta)$ and $\Delta(r)$, and taking $r_0$ large enough and choose $\beta$ small enough accordingly.  

\section{Locality analysis}\label{sec:local}

\subsection{Decay of generators with the distance}

We here focus on the $(k,l)$-local Hamiltonians defined in the main text. For such Hamiltonians, given the so-called Lieb-Robinson velocity
$v_{\text{LR}} := \max_{u\in \Lambda} \sum_{ Z\ni u} \vert Z \vert \norm{h_Z}_\infty \le hkl \equiv J$,
it is well-known that, for any operator $A^a$ supported on site $a \in \Lambda$ \cite{lieb1972finite,hastings2004lieb,nachtergaele2006lieb,HastingsKoma2006,Haah_2021}:
\begin{equation}\label{eq:LiebRobinson}
	\norm{e^{-iHt}A^a e^{iHt}-e^{-iH_{B_a(r)}t}A^a e^{iH_{B_a(r)}t}}_\infty \le \norm{A^a}_\infty \frac{(2 v_{\text{LR}}\vert t \vert)^r}{r!}\le \norm{A^a}_\infty \frac{(2J \vert t \vert)^r}{r!} \,,
\end{equation}
where $B_a(r)$ denotes the ball of radius $r$ with respect to the graph distance centered at $a$, so that $ H_{B_b(r)}$ consists of the Hamiltonian terms on a ball of radius $r$ centered at the support of $A^{a}$. In what follows, in order to ease notation, we will be identifying $H_r \equiv H_{B_a(r)}$.

We aim to show that there is a sequence of Lindbladians with jump operators with support on a region of radius $r$ such that, for a fast-decaying function $\zeta(r)$,
\begin{align}
	&\|\mathcal{L}_a^{\beta,r\dagger}-\mathcal{L}_{a}^{\beta,r-1\dagger }\|_{\infty\to\infty}\le \zeta(r)\,.\end{align}
The Lindbladian $\mathcal{L}_a^{\beta,r\dagger}$ is the same as the one defined in the main text, but where the jump operators are instead
\begin{align}
	A^{a,\alpha}_r(\omega):=\frac{1}{\sqrt{2\pi}}\int_{-\infty}^{\infty} e^{iH_rt}A^{a,\alpha} e^{-iH_r t}e^{-i\omega t}\,f(t)\,dt\,,
\end{align}
and the coherent part is 
\begin{align}
	B^\beta_{a,\alpha,r} &\equiv \int_{-\infty}^\infty b_1(t) e^{-i\beta H_r t} \int_{-\infty}^\infty b_2(t') e^{i\beta H_r t'}A^{a,\alpha} e^{-2i\beta H_r t'}A^{a,\alpha} e^{i\beta H_r t'}dt' e^{i\beta H_r t}dt \\ & = \int_{-\infty}^\infty  \int_{-\infty}^\infty b_1(t) b_2(t') A_r^{a,\alpha} (t-t') A_r^{a,\alpha}(t+t') dt' dt\,.
\end{align}

First, notice that, given the Lieb-Robinson bound, 
\begin{align}
	\norm{A^{a,\alpha}_r(\omega) - A^{a,\alpha}_{r+1}(\omega)}_\infty &\le  \frac{\norm{A^{a,\alpha}}_\infty}{\sqrt{2 \pi}} \int_{-t_0}^{t_0} f(t) \frac{(2 J \vert t \vert )^r}{r!} \text{d} t +  \frac{\norm{A^{a,\alpha}}_\infty}{\sqrt{\pi/2}} \int_{t_0}^\infty  f(t) \text{d} t
	\\ & \le \norm{A^{a,\alpha}}_\infty \sqrt{\beta / \sqrt{2 \pi} } \frac{(2 J e  \vert t_0 \vert )^r}{r^{r}} +  \norm{A^{a,\alpha}}_\infty \frac{\beta^{3/2}}{(2 \pi^3)^{1/4} t_0} e^{-\frac{t_0^2}{\beta^2}}\,,
\end{align}
where in the last line we used that $r! \ge \frac{r^r}{e^{r-1}}$ and the upper bound on the complementary error function $\text{Erfc}(x) \le \frac{e^{-x^2}}{\sqrt{\pi}x}$.
Let us choose $t_0 = \frac{r g(\beta J)} {2J e}$ with $g(x)=\frac{\sqrt{x}}{1+\sqrt{x}}$, so that
\begin{align}
	\norm{A^{a,\alpha}_r(\omega) - A^{a,\alpha}_{r+1}(\omega)}_\infty &\le \sqrt{\beta} \norm{A^{a,\alpha}}_\infty \left ( \frac{g(\beta J)^r}{ (2 \pi)^{1/4} } +  \frac{  \beta J e }{( \pi/2)^{3/4} r g(\beta J)} e^{-\frac{r^2 g(\beta J)^2}{(2 \beta J e )^2}} \right) \\ & \equiv \frac{\sqrt{\beta}}{ 2 (2 \pi)^{3/4} } \norm{A^{a,\alpha}}_\infty  \xi_1(r,\beta J)
\end{align}
Thus, considering that $\norm{\int_{-\infty}^{\infty} \gamma (\omega) A^{a,\alpha}_r(\omega)  \text{d} \omega}_\infty \le \norm{A^{a,\alpha}}_\infty (2 \pi)^{3/4} /\sqrt{\beta}  $, we have that 
\begin{align}
	\norm{  \mathcal{L}^{a,\alpha}_{d,r}-\mathcal{L}^{a,\alpha}_{d,r+1} }_{\infty\to\infty} \le  \norm{A^{a,\alpha}}_\infty^2  \xi_1(r,\beta J)\,.
\end{align}

For the coherent part, we have that 
\begin{align}
	\norm{  B^\beta_{a,\alpha,r}- B^\beta_{a,\alpha,r+1} }_\infty \le \int_{-\infty}^\infty  \int_{-\infty}^\infty \vert b_1(t) b_2(t') \vert \norm{ A_r^{a,\alpha} (t-t') A_r^{a,\alpha}(t+t') -  A_{r+1}^{a,\alpha} (t-t') A_{r+1}^{a,\alpha}(t+t') }_\infty dt' dt .
\end{align}
We divide the integral with a different cut-off $t_0$ as $$ \int_{-\infty}^\infty  \int_{-\infty}^\infty = \int_{-t_0}^{t_0}  \int_{-t_0}^{t_0} + \int_{t_0}^\infty  \int_{-\infty}^\infty + \int_{-\infty}^\infty\int_{t_0}^\infty + \int_{t_0}^\infty \int_{t_0}^\infty +\quad ...$$ so that
\begin{align}
	\int_{-t_0}^{t_0}  \int_{-t_0}^{t_0}   \vert b_1(t) b_2(t') \vert \norm{ A_r^{a,\alpha} (t-t') A_r^{a,\alpha}(t+t') -  A_{r+1}^{a,\alpha} (t-t') A_{r+1}^{a,\alpha}(t+t') }_\infty dt' dt \le \frac{e^{1/8}}{2 \sqrt{2}} \frac{(4 \beta J   e \vert t_0 \vert )^r}{r^{r}} \norm{A^{a,\alpha}}_\infty^2,
\end{align}
where we used the Lieb-Robinson bound and that $\int_{-\infty}^\infty  \int_{-\infty}^\infty \vert b_1(t) b_2(t') \vert dt' dt \le\frac{e^{1/8}}{4 \sqrt{2}}$. For the other contributions, where long times contribute, we simply bound $\norm{ A_r^{a,\alpha} (t-t') A_r^{a,\alpha}(t+t') -  A_{r+1}^{a,\alpha} (t-t') A_{r+1}^{a,\alpha}(t+t') }_\infty \le 2 \norm{A^{a,\alpha}}_\infty^2 $. Considering the tail bounds
\begin{align}
	&\int_{t_0}^\infty \vert b_1(t) \vert \text{d}t \le \frac{e^{1/8}}{2 \sqrt{2}}\left(e^{\frac{\pi^2}{2}-2 \pi t_0}+2 \text{Erfc}(\sqrt{2} t_0)  \right)  \\
	& \int_{t_0}^\infty \vert b_2(t) \vert \text{d}t \le \frac{\text{Erfc}(2 t_0)}{8 \pi}, 
\end{align}
we can upper bound the contribution of all the tail terms in the integral as
\begin{align}
	\left(  \int_{-\infty}^\infty\int_{t_0}^\infty + \int_{t_0}^\infty \int_{t_0}^\infty +... \right) \vert b_1(t) b_2(t') \vert dt' dt \le
	\frac{e^{1/8} }{ 2 \pi \sqrt{2}}\left(e^{\frac{\pi^2}{2}-2 \pi t_0}+4 \frac{e^{-2 t_0^2}}{\sqrt{2 \pi} t_0}   \right).
\end{align}
Now similarly choosing $t_0 = g(\beta J) r/ 4 \beta J e   $ we obtain

\begin{align}
	\norm{  B^\beta_{a,\alpha,r}- B^\beta_{a,\alpha,r+1} }_\infty & \le \frac{e^{1/8} \norm{A^{a,\alpha}}_\infty^2}{2 \sqrt{2}} \left ( g(\beta J)^r + e^{\frac{\pi^2}{2}- \frac{\pi g(\beta J)}{2 \beta J e } r  }+\frac{8 \beta J e  }{\sqrt{2 \pi} r g(\beta J) } e^{- \frac{(r g(\beta J))^2}{2 ( \beta J  e )^2}}  \right) \\& \equiv \frac{\norm{A^{a,\alpha}}_\infty^2}{2} \xi_2(r,\beta J).
\end{align}

Assuming that $\norm{A^{a,\alpha}}_\infty \le 1$, this means that we can take the function $ \zeta(r)$ to be
\begin{align}
	\zeta(r) = \xi_1(r,\beta J) +\xi_2(r,\beta J)\,,
\end{align}
which is exponentially decaying in $r$ and vanishes as $\beta \rightarrow 0$.
Denoting $\upsilon:=\beta J/g(\beta J)=\sqrt{\beta J}(1+\sqrt{\beta J})$ and taking $r\ge 1$, a crude bound on $\zeta$ is
\begin{align} \label{eq:3terms}
	\zeta(r)&\le 7g(\beta J)^r+23\upsilon e^{-\frac{r^2 }{4 e^2 \upsilon^2}}+112 e^{-\frac{\pi r}{e \upsilon }}.
\end{align}
As long as $\beta J\le 1/200$, it can be verified numerically that the first term in Eq. \eqref{eq:3terms} dominates (it is enough to check for $r=1$), so that the other two contribute at most as $7g(\beta J)^r$. This means we have
\begin{align}
	\zeta(r)&\le 14 g(\beta J)^r,
\end{align}
and hence,
\begin{align}
	\Delta(r_0)=\sum_{r\ge r_0}\zeta(r)&\le \frac{14 g(\beta J)^{r_0}}{1-g(\beta J)}.
\end{align}



\subsection{Distance to the depolarizing limit}
We now need to find a function such that
\begin{align}
	\|\mathcal{L}_a^{\beta \dagger}-\mathcal{L}_{a}^{0 \dagger }\|_{\infty\to\infty}\le \eta(\beta)\,.\end{align}
We proceed as above by treating the dissipative and the coherent parts separately.
First notice that the jump operators are such that

\begin{align}
	\norm{A^{a,\alpha}(\omega)- \tilde{f}(\omega)A^{a,\alpha}}_\infty &\le \frac{1}{\sqrt{2\pi}}\int_{-\infty}^{\infty} \norm{A^{a,\alpha}(-t)-A^{a,\alpha} }_\infty \,f(t)\,dt\qquad 
	\\ & \le \frac{\norm{[H,A^{a,\alpha}]}_\infty}{\sqrt{2\pi}}\int_{-\infty}^{\infty}  \vert t \vert  f(t) \text{d}t
	\\ & \le \frac{1}{2^{\frac{1}{4}} \pi^{\frac{3}{4}}} 2 \beta^{3/2} J \norm{A^{a,\alpha}}_\infty,
\end{align}
where $\tilde{f}(\omega)= \frac{1}{\sqrt{2 \pi}} \int^{\infty}_{-\infty}f(t) e^{-i \omega t}=\frac{\sqrt{\beta}}{(2 \pi)^{1/4}} e^{\frac{-\beta^2 \omega^2}{4}}$, and
where we used the fact that $\norm{A(t)-A}_\infty\le \vert t \vert \norm{[H,A]}_\infty\le 2J\vert t \vert \norm{A}_\infty $. Thus, by the same argument as above, we have, for the dissipative part,
\begin{align}
	\|\mathcal{L}_{a,d}^{\beta \dagger}-\mathcal{L}_{a,d}^{0 \dagger }\|_{\infty\to\infty}\le \frac{2^{5/2}}{\sqrt{3 \pi}e^{1/6}} \beta J.
\end{align}

Now for the coherent part, we have that, for $\lambda =\int_{-\infty}^\infty  \int_{-\infty}^\infty  b_1(t) b_2(t')  dt' dt$, since $(A^{a,\alpha})^2=\mathbb{I}$,
\begin{align}
	\norm{  B^\beta_{a,\alpha}- \lambda I }_\infty &\le \int_{-\infty}^\infty  \int_{-\infty}^\infty \vert b_1(t) b_2(t') \vert {\norm{A^{a,\alpha} (-\beta t') A^{a,\alpha}(\beta t')- I }_\infty}\,dt' dt
	\\ &
	\le 4 \beta J \norm{A^{a,\alpha}}_\infty^2 \int_{-\infty}^\infty  \int_{-\infty}^\infty \vert b_1(t) b_2(t') \vert \vert t'\vert dt' dt = \frac{e^{\frac{1}{8}}}{2 \sqrt{2 \pi}} \beta J \norm{A^{a,\alpha}}_\infty^2.
\end{align}
This means that we can simply take
\begin{align}
	\eta(\beta) &:=  2 \beta J.
\end{align}

\subsection{Long range Hamiltonians}\label{app:longrange}

Here we show how the previous estimates on $\eta(\beta)$ and $\Delta(r)$ generalize to long range Hamiltonians. Their definition implies that for $A^{a,\alpha}$ being single-body Paulis, $\norm{[H,A^{a,\alpha}]} \le 2 g$, and so we can automatically establish that for long range systems we can choose $   \eta(\beta):=  \beta g$. For estimating $\Delta(r)$, we resort to existing Lieb-Robinson bounds for those models \cite{Kuwahara2019LongRange}, which read
\begin{equation}\label{eq:LiebRobinson}
	\norm{e^{-iHt}A^u e^{iHt}-e^{-iH_{B_u(r)}t}A^u e^{iH_{B_u(r)}t}}_\infty \le \norm{A^u}_\infty C_H  \vert t \vert^{D+1} \left (r - v_{\text{LR}} t \right)^{-\nu+D}.
\end{equation}
The constants $C_H,v_{\text{LR}}$ now do not have explicit expressions, but depend only on $D, g_0$ and $\nu$ (see \cite{Kuwahara2019LongRange}). 

Repeating the proof above with this bound instead, and choosing $v_{\text{LR}} t_0=\sqrt{\beta g} r $ for $\beta g <1$, we obtain, for the dissipative part

\begin{align}
	\norm{  \mathcal{L}^{a,\alpha}_{d,r}-\mathcal{L}^{a,\alpha}_{d,r+1} }_{\infty\to\infty} \le  \norm{A^{a,\alpha}}_\infty^2  \left(2 \sqrt{2 \pi} \frac{C_H (\sqrt{\beta g}/ v_{\text{LR}})^{D+1}}{(r(1-\sqrt{\beta g}))^{\nu-2D-1}} +  \frac{4 \sqrt{2 \beta}  v_{\text{LR}}}{ \sqrt{ g} r} e^{-\frac{g r^2}{4\beta  v_{\text{LR}}^2}} \right),
\end{align}
and similarly for the coherent part, choosing instead $v_{\text{LR}} t_0=\sqrt{ g/\beta } r $
\begin{align}
	\norm{  B^\beta_{a,\alpha,r}- B^\beta_{a,\alpha,r+1} }_\infty & \le \frac{e^{1/8} \norm{A^{a,\alpha}}_\infty^2}{2 \sqrt{2}} \left ( \frac{C_H (\sqrt{\beta g}/ v_{\text{LR}})^{D+1}}{(r(1-\sqrt{\beta g}))^{\nu-2D-1}} + e^{\frac{\pi^2}{2}- \frac{\pi\sqrt{g}}{ v_{\text{LR}}\sqrt{\beta}} r  }+\frac{4 v_{\text{LR}} \sqrt{\beta}  }{\sqrt{2 \pi g}  r} e^{- \frac{2 g r^2}{\beta v_{\text{LR}}^2 }}  \right).
\end{align}
This means that, assuming $\beta g <1/4$, we can choose, for some constant $C'_H$,
\begin{align}
	\zeta(r)&\le C'_H \frac{(\beta g)^{\frac{D+1}{2}}}{r^{\nu-2D-1}},
\end{align}
so that, equivalently, assuming $\nu>2D+2$, there is a contant $K>0$ such that we can choose
\begin{align}
	\Delta(r_0) := K  (\beta g)^{\frac{D+1}{2}} \frac{1}{r_0^{\nu -2D-2 }}.
\end{align}


\section{Partition function estimation}\label{sec:partition_function}

The standard algorithm to estimate partition functions given access to (approximate) samples from Gibbs states from a temperature range $\beta\in[\beta_{\min},\beta_{\max}]$ assumes that we know the value of the partition function at $\beta_{\min}$. Then, for a so-called annealing schedule of length $l$ given by inverse temperatures $\beta_1=\beta_{\min}<\beta_1<\cdots<\beta_l=\beta_{\max}$, we sample from $\sigma_{\beta_k}$ and estimate the expectation value of the observable $e^{(\beta_{i}-\beta_{i+1})H}$ on it. It is then not difficult to see that $\tr{e^{(\beta_{i}-\beta_{i+1})H}\sigma_{\beta_i}}=\tfrac{Z_{\beta_{i+1}}}{Z_{\beta_i}}$, so we obtain an estimate of the partition function $Z_{\beta_{\max}}$ by considering the telescopic product
\begin{align}
	Z_{\beta_{\max}}=Z_{\beta_{\min}}\prod\limits_{i=1}^{l-1}\tfrac{Z_{\beta_{i+1}}}{Z_{\beta_i}}.
\end{align}
It is then important to design the schedule in a way that we can simultaneously ensure that it is as short as possible (minimize $l$) while being such that we do not need too many samples to estimate the ratio for each $\beta_k$.

In the quantum case, to the best of our knowledge, the only result on annealing schedules is that of~\cite{poulin_sampling_2009}, which considers schedules with $\beta_{k+1}-\beta_k=\mathcal{O}(\|H\|_\infty^{-1})$ and $l=\mathcal{O}((\beta_{\max}-\beta_{\min})\|H\|_\infty)$. However, as it is known in the classical case, this schedule can be quadratically worse than the optimal one~\cite{Stefankovic2007}. Another challenge the quantum version of partition function estimation poses is that, in general, we cannot directly measure the observable $e^{(\beta_{k}-\beta_{k+1})H}$, as required by the annealing method. The work of~\cite{poulin_sampling_2009} already gave a quantum algorithm based on quantum phase estimation to deal with this issue. 

Here we take a simplified and self-contained route by using block encodings~\cite{Gilyn2019}. We note that several other papers used similar ideas to obtain quantum algorithms to estimate partition functions. However they focused on the regime where the sample complexity is exponential in the system size~\cite{randomized_matrix_signal,Jackson2023,Chowdhury2021}.

We now assume we have access to a block encoding of $\frac{H}{\|H\|_\infty}$~\cite{Gilyn2019}. Furthermore, we assume we are given a so-called balanced cooling schedule~\cite{Stefankovic2007}, i.e., a sequence of length $l$ with $\beta_{\min}=\beta_1 \leq\ldots\leq\beta_l=\beta_{\max}$ that satisfies for all $1\leq i\leq l-1$ and a parameter $B>0$ such that

\begin{align}\label{equ:good_schedule}
	\frac{Z_{\beta_{i}}}{Z_{\beta_{i+1}}} \leq B
\end{align}
The reason for this requirement is that, as we will see later, given access to copies of the $\sigma_{\beta}$ and the block encoding, it is possible to sample from a random variable $X_i$ that satisfies
\begin{align}\label{equ:variance_partition_function}
	\mathbb{E}\left(X_{i}\right)=\frac{Z_{\beta_{i+1}}}{Z_{\beta_{i}}},\quad \frac{\mathbb{E}\left(X_{i}^{2}\right)}{\mathbb{E}\left(X_{i}\right)^{2}} = \frac{Z_{\beta_{i}}}{Z_{\beta_{i+1}}} \leq B\,.
\end{align}
The conditions in Equation~\eqref{equ:variance_partition_function} can be combined with a standard result by Dyer and Frieze~\cite{dyer1991computing} to show how to obtain an estimator for the partition function $Z_{\beta_{\max}}$ by sampling from the $X_i$.
\begin{lemma}\label{lemma:frieze}
	Let $X_1,\ldots,X_l$ be independent random variables with 
	\begin{align}
		\frac{\mathbb{E}\left(X_{i}^{2}\right)}{\mathbb{E}\left(X_{i}\right)^{2}}\leq B
	\end{align}
	for all $1\leq i\leq l$. Let $\hat{X}=X_1X_2\cdots X_l$, $S_i$ be the average of $16Bl\epsilon^{-2}$ independent samples from $X_i$ and $\hat{S}=S_1S_2\cdots S_l$. Then
	\begin{align}
		\mathbb{P}((1-\epsilon)\mathbb{E}(\hat{X})\leq\hat{S}\leq (1+\epsilon)\mathbb{E}(\hat{X}))\geq 3/4.
	\end{align}
\end{lemma}

Before we show how to sample from $X_{i}$, let us recall some well-known facts about block encodings of the exponential function.
\begin{lemma}[Block encoding of the exponential function~\cite{vanApeldoorn2020}]
	Let $H\geq 0$ be a Hamiltonian and $U$ a block encoding of $\tfrac{H}{\|H\|_\infty}$. For $\beta,\epsilon\geq0$ we can obtain an $\epsilon$ block encoding $U_\beta$ of $e^{-\beta H}$, i.e. a unitary satisfying
	\begin{align}\label{equ:definition_block}
		\|(I\otimes \bra{0})U_\beta(I\otimes \ket{0})-e^{-\beta H}\|_\infty\leq \epsilon
	\end{align}  
	from 
	\begin{align}
		\mathcal{O}(\beta\|H\|_\infty\log(\epsilon^{-1}))
	\end{align}
	uses of the block encoding of $H/\|H\|_\infty$.
\end{lemma}
This fact can be readily applied to sample from $X_i$:
\begin{lemma}\label{lemma:effect_block_encoding}
	For some $\beta_1,\beta_2>0$, let $U_{\beta_1}$ be an $\epsilon$ block encoding of $e^{-\beta_1 H}$ as in~\eqref{equ:definition_block}. Then
	\begin{align}
		\left|\tr{I \otimes \ketbra{0}{0}\left(U_{\beta_1}(\sigma_{\beta_2}\otimes\ketbra{0}{0})U_{\beta_1}^\dagger\right)} - \frac{Z_{2\beta_{1}+\beta_{2}}}{Z_{\beta_{2}}}\right| \leq 2\epsilon
	\end{align}
\end{lemma}
\begin{proof}
	If we had an exact block encoding of $e^{-\beta_{1} H}$, denote it by $V_{\beta_{1}}$, then
	\begin{align}
		&\tr{I \otimes \ketbra{0}{0}\left(V_{\beta_1}(\sigma_{\beta_2}\otimes\ketbra{0}{0})V_{\beta_1}^\dagger\right)}\\
		&\qquad =\tr{I \otimes \ketbra{0}{0}\left(V_{\beta_1}(\sigma_{\beta_2}\otimes\ketbra{0}{0})V_{\beta_1}^\dagger\right) I \otimes \ketbra{0}{0}}\\
		&\qquad =\tr{e^{-\beta_1 H}\otimes \ketbra{0}{0}\left(\frac{e^{-\beta_2H}}{\tr{e^{-\beta_2 H}}}\otimes \ketbra{0}{0}\right)e^{-\beta_1 H}\otimes \ketbra{0}{0}}\\
		&\qquad =\frac{Z_{2\beta_{1}+\beta_{2}}}{Z_{\beta_{2}}}.
	\end{align}
	The lemma then follows from the fact that $U_{\beta_1}$ approximates $V_{\beta_1}$ up to $\epsilon$ in operator norm, and thus also in diamond norm.
\end{proof}
From an approximate block encoding $U_i$ of $e^{-\tfrac{(\beta_{i+1}-\beta_{i})H}{2}}$, and given a copy of $\sigma_{\beta_i}$, we let $X_i=1$ if we measure $0$ in the auxilliary system after applying the block encoding $U_i$ to $\sigma_{\beta_i}\otimes\ketbra{0}{0}$ and $X_i=0$ else.
We then have:
\begin{cor}
	For $X_i$ defined as above, an annealing schedule $\beta_1<\beta_2<\cdots<\beta_l$ satisfying \eqref{equ:good_schedule} and $\epsilon\leq \frac{B^{-1}}{4}$  we have
	\begin{align}
		\left|\mathbb{E}(X_i)-\frac{Z_{\beta_{i+1}}}{Z_{\beta_{i}}}\right|\leq 2\epsilon,\quad \frac{\mathbb{E}\left(X_{i}^{2}\right)}{\mathbb{E}\left(X_{i}\right)^{2}}\leq 2B.
	\end{align}
\end{cor}
\begin{proof}
	The claim on the expectation value of $X_i$ directly follows from Lemma~\ref{lemma:effect_block_encoding} and the bound on the relative variance by combining \eqref{equ:good_schedule} with the assumption that $\epsilon\leq \frac{B^{-1}}{4}$. Indeed, we have that 
	\begin{align}
		\frac{\mathbb{E}\left(X_{i}^{2}\right)}{\mathbb{E}\left(X_{i}\right)^{2}}=\frac{1}{\mathbb{E}\left(X_{i}\right)},
	\end{align}
	as $X_i$ are Bernoulli random variables. From our assumptions on the schedule and $\epsilon$ we obtain
	\begin{align}
		\mathbb{E}\left(X_{i}\right)\geq \frac{Z_{\beta_{i+1}}}{Z_{\beta_{i}}}-2\epsilon\geq B^{-1}-2\epsilon\geq \frac{1}{2B}.
	\end{align}
\end{proof}
Combining all of these statements we obtain the final result.
\begin{thm}[Partition function estimation]
	For a Hamiltonian $H$, $\beta_{\min}=0$ and $\beta_{\max}\leq\beta^*$ for which rapid mixing holds, let $\beta_{\min}=\beta_1<\cdots<\beta_l=\beta_{\max}$ be an annealing schedule satisfying \eqref{equ:good_schedule} for some $B>0$. Then we can obtain an estimate $\hat{Z}_{\beta_{\max}}$ of $Z_{\beta_{\max}}$ up to a relative error $\epsilon>0$ with probability of success at least $3/4$ from 
	\begin{align}
		\widetilde{\mathcal{O}}(n l^2B\epsilon^{-2})
	\end{align}
	uses of a block encoding of $\tfrac{H}{\|H\|_\infty}$ and Hamiltonian evolution time for $H$, where $\widetilde{\mathcal{O}}$ hides polylogarithmic terms in $\epsilon, l$ and $B$. In particular, for short-range lattice Hamiltonians there is a choice of parameters leading to an algorithm with runtime
	\begin{align}
		\widetilde{\mathcal{O}}(n^3\epsilon^{-2})\,.
	\end{align}
	and for long-range evolutions we have
	\begin{align}
		\widetilde{\mathcal{O}}(n^4\epsilon^{-2})\,.
	\end{align}
\end{thm}
\begin{proof}
	As $\beta^*$ is of constant order, so are all $\beta_{i+1}-\beta_i$. Thus, we can generate a $\mathcal{O}(\frac{\epsilon}{Bl})$ approximate block encoding of $e^{-\tfrac{\beta_{i+1}-\beta_{i}}{2}H}$ from $\widetilde{\mathcal{O}}(\|H\|_\infty)$ uses from the block encoding of $\tfrac{H}{\|H\|_\infty}$. Furthermore, as claimed in the main text, we can prepare a copy of $\sigma_{\beta_i}$ with access to $\widetilde{\mathcal{O}}(n)$ Hamiltonian evolution time of $H$. This can be generated with $\widetilde{O}(n\|H\|_\infty)$ uses of the block encoding of $H/\|H\|_\infty$ and, in the case of short-range lattice Hamiltonians, $\widetilde{O}(n)$~\cite{Haah_2021}. Thus, we conclude that to generate a sample $\widetilde{X}_i$ that is $\mathcal{O}(\tfrac{\epsilon^2}{lB})$ close in total variation distance from $X_i$ defined before, we need $\widetilde{\mathcal{O}}(n+\|H\|_\infty)$ uses of the block encoding.
	In particular, the empirical average of $\mathcal{O}(\epsilon^{-2}Bl)$ such random variables is indistinguishable with success probability $\geq\frac{2}{3}$ from that of the $X_i$. It then follows from Lemma~\ref{lemma:frieze} that we can take the empirical average of $16Bl\epsilon^{-2}$ such samples to obtain the estimator $\hat{Z}_{\beta_{\max}}$ with the claimed relative error. By multiplying the individual sample complexity for each $i$ with the cost to generate one sample, we get the claimed complexity. 
	
	Let us now discuss how to choose the parameters to obtain the claimed runtime for lattice Hamiltonians. Note that it follows from an application of the Golden-Thompson inequality and H\"older inequalities that it is possible to pick $\beta_{i+1}-\beta_i=\mathcal{O}(1/\|H\|_\infty)$ and achieve $B=\mathcal{O}(1)$. This would then yield a schedule with $l=\|H\|_\infty$. As for lattice Hamiltonians (short or long range) we have $\|H\|_\infty=\mathcal{O}(n)$, this gives a total complexity of $\widetilde{\mathcal{O}}(n^3\epsilon^{-2})$ to estimate the partition function to relative precision. 
\end{proof}

\end{document}